\documentclass[11pt]{llncs}
\usepackage{latexsym}
\usepackage{amssymb}
\usepackage{amsmath,amstext}
 \usepackage[pdftex]{graphicx}
\usepackage{owna4}

\makeatletter
\def\@fnsymbol#1{\ensuremath{\ifcase#1\or\star\or \dagger\or \ddagger\or
  {\star\star}\or {\star\star\star}\or
   \mathchar "278\or \mathchar "27B\or \|\or **\or \dagger\dagger
   \or \ddagger\ddagger \else\@ctrerr\fi}}
\makeatother

\makeatletter
\def~{\ifmmode\;\else\penalty\@M\ \fi}
\def\@setmcodes#1#2#3{{\count0=#1 \count1=#3
  \loop \global\mathcode\count0=\count1 \ifnum \count0<#2
  \advance\count0 by1 \advance\count1 by1 \repeat}}
\DeclareSymbolFont{italic}{OT1}{\rmdefault}{m}{it}
\let\mathit\undefined
\DeclareSymbolFontAlphabet{\mathit}{italic}
\edef\@tempa{\hexnumber@\symitalic}
\@setmcodes{`A}{`Z}{"7\@tempa41}
\@setmcodes{`a}{`z}{"7\@tempa61}
\makeatother
\newdimen\asmindent     
\asmindent=\parindent
\newcount\asmi
\def\inc{\global\advance\asmi by 1}
\def\dec{\global\advance\asmi by-1}
\def\nl{{}$\par\hangindent\asmi em
  \noindent\kern\asmi em\ignorespaces$} 
\def\asmskip{{}$\par\smallskip\hangindent\asmi em
  \noindent\kern\asmi em\ignorespaces$} 
\def\asm{\global\asmi=0 
 \def\+{\inc\nl}
 \def\-{\dec\nl}
 \def\\{\nl}
 \begin{trivlist}\item[]\leftskip=\asmindent\relax$}
\def\endasm{$\end{trivlist}}
\def\asmarray{\begin{array}[t]{@{}l@{\;}l@{\;}l@{}}}
\def\endasmarray{\end{array}}
\newcount\asmii
\def\subasm{\vtop\bgroup\asmii=0\normalbaselines
 \def\nl##1{$\egroup\advance\asmii by##1\relax\hbox\bgroup\hskip\asmii em$}
 \def\\{\nl{0}}
 \def\+{\nl{1}}
 \def\-{\nl{-1}}
 \hbox\bgroup\hskip\asmii em$}
\def\endsubasm{$\egroup\egroup}
\def\ASM#1{\hbox{\sc#1}}        

\def\AND     {\mathrel{\mathbf{and}}}

\def\CHOOSE  {\mathrel{\mathbf{choose}}}

\def\DO      {\mathrel{\mathbf{do}}}
\def\ELSE    {\mathrel{\mathbf{else}}}

\def\FORALL  {\mathrel{\mathbf{forall}}}
\def\FORSOME  {\mathrel{\mathbf{forsome}}}

\def\THEREIS  {\mathrel{\mathbf{thereIs}}}

\def\IF      {\mathrel{\mathbf{if}}}

\def\IN      {\mathrel{\mathbf{in}}}
\def\LET     {\mathrel{\mathbf{let}}}

\def\NOT     {\mathrel{\mathbf{not}}}

\def\OR      {\mathrel{\mathbf{or}}}
\def\PAR     {\mathrel{\mathbf{par}}}
\def\SEQ     {\mathrel{\mathbf{seq}}}
\def\SKIP    {\mathrel{\mathbf{skip}}}
\def\THEN    {\mathrel{\mathbf{then}}}
\def\WHERE   {\mathrel{\mathbf{where}}}

\def\UNTIL   {\mathrel{\mathbf{until}}}

\def\WITH    {\mathrel{\mathbf{with}}}

\def\SEQ    {\mathrel{\mathbf{seq}}}
\def\ITERATE    {\mathrel{\mathbf{iterate}}}

\makeatletter
\def\enumerate{%
  \ifnum \@enumdepth >\thr@@\@toodeep\else
    \advance\@enumdepth\@ne
    \edef\@enumctr{enum\romannumeral\the\@enumdepth}%
      \expandafter
      \list
        \csname label\@enumctr\endcsname
        {\usecounter\@enumctr\def\makelabel##1{\hss\llap{##1}}
         \itemsep 0pt\parskip 0pt\parsep 0pt\topsep\smallskipamount}%
  \fi}
\def\itemize{%
  \ifnum \@itemdepth >\thr@@\@toodeep\else
    \advance\@itemdepth\@ne
    \edef\@itemitem{labelitem\romannumeral\the\@itemdepth}%
    \expandafter
    \list
      \csname\@itemitem\endcsname
      {\def\makelabel##1{\hss\llap{##1}}
       \itemsep 0pt\parskip 0pt\parsep 0pt\topsep\smallskipamount}%
  \fi}
\makeatother
\title{Serialisable Multi-Level Transaction Control:\\ A Specification
  and Verification\thanks{The research reported in this paper results
    from the project \textit{Behavioural Theory and Logics for
      Distributed Adaptive Systems} supported by the \textbf{Austrian
      Science Fund (FWF): [P26452-N15]}. The first author, Humboldt
  research prize awardee in 2007/08, gratefully acknowledges partial
  support by a renewed research grant from the Alexander von Humboldt
  Foundation in 2014.}\thanks{The final publication is available at Elsevier via https://doi.org/10.1016/j.scico.2016.03.008.}\thanks{\copyright 2017. This manuscript version is made available under the CC-BY-NC-ND 4.0 license http://creativecommons.org/licenses/by-nc-nd/4.0/.}}

\author{Egon B{\"o}rger\inst{1}, Klaus-Dieter Schewe\inst{2}, Qing Wang\inst{3}}

\institute{Universit\`{a} di Pisa, Dipartimento di Informatica, I-56125 Pisa, Italy\\
\email{boerger@di.unipi.it} \and
Software Competence Centre Hagenberg, A-4232 Hagenberg, Austria\\
\email{kdschewe@acm.org} \and
Research School of Computer Science, Australian National University, Australia\\
\email{qing.wang@anu.edu.au}}

\begin{document}

\maketitle

\begin{abstract}

We define a programming language independent controller $\ASM{TaCtl}$
for multi-level transactions and an operator $TA$, which when applied
to concurrent programs with multi-level shared locations containing
hierarchically structured complex values, turns their behavior with
respect to some abstract termination criterion into a transactional
behavior. We prove the correctness property that concurrent runs under
the transaction controller are serialisable, assuming an \emph{Inverse
  Operation Postulate} to guarantee recoverability. For its
applicability to a wide range of programs we specify the transaction
controller $\ASM{TaCtl}$ and the operator $TA$ in terms of Abstract
State Machines (ASMs). This allows us to model concurrent updates at
different levels of nested locations in a precise yet simple manner,
namely in terms of partial ASM updates. It also provides the
possibility to use the controller $\ASM{TaCtl}$ and the operator $TA$
as a plug-in when specifying concurrent system components in terms of
sequential ASMs.

\end{abstract}

\section{Introduction}

This paper is about the use of generalized multi-level transactions as a means to
control the consistency of concurrent access of programs to shared
locations, which may contain hierarchically structured complex values,
and to avoid that values stored at these locations are changed almost
randomly. According to Beeri, Bernstein and Goodman \cite{BeBeGo89} most real systems with shared data have multiple levels, where each level has its own view of the data and its own set of operations, such that operations on one level may be conflict-free, while they require conflicting lower-level operations.

A multi-level {\em transaction controller} interacts with
concurrently running programs (i.e., sequential components of an
asynchronous system) to control whether access to a possibly
structured shared location can be granted or not, thus ensuring a
certain form of consistency for these locations. This includes in
particular the resolution of low-level conflicts by higher-level
updates as provided by multi-level
transactions~\cite{BeBeGo89,Weikum86,Weikum91} in distributed
databases~\cite{BerGoo81,OezVal94}. A commonly accepted consistency
criterion is that the joint behavior of all transactions (i.e.,
programs running under transactional control) with respect to the
shared locations is equivalent to a serial execution of those
programs. Serialisability guarantees that each transaction can be
specified independently from the transaction controller, as if it had
exclusive access to the shared locations.

It is expensive and cumbersome to specify transactional behavior and
prove its correctness again and again for components of the great
number of concurrent systems. Our goal is to define once and for all
an abstract (i.e. programming language independent) transaction
controller $\ASM{TaCtl}$ which can simply be ``plugged in'' to turn
the behavior of concurrent programs (i.e., components~$M$ of any given
asynchronous system $\cal M$) into a transactional one.  This involves
to also define an operator~$TA(\bullet,\ASM{TaCtl})$ that transforms a
program $M$ into a new one $TA(M,\ASM{TaCtl})$, by means of which the
programs~$M$ are forced to listen to the controller $\ASM{TaCtl}$ when
trying to access shared locations. To guarantee recoverability where
needed we use an \emph{Inverse Operation Postulate}
(Sect.\ref{sect:recovery}) for component machines $M$; its satisfaction
is a usage condition for submitting $M$ to the transaction controller.

For the sake of generality we define the operator and the controller
in terms of Abstract State Machines (ASMs), which can be read and
understood as pseudo-code so that $\ASM{TaCtl}$ and the operator $TA$
can be applied to code written in any programming language (to be
precise: whose programs come with a notion of single step, the level
where our controller imposes shared memory access constraints to
guarantee transactional code behavior). The precise semantics
underlying the pseudo-code interpretation of ASMs (for which we refer
the reader to~\cite{BoeSta03}) allows us to mathematically prove the
correctness of our controller and operator. 

Furthermore, we generalize the strictly hierarchical view of multiple
levels by using the partial update concept for ASMs developed
in~\cite{gurevich:jucs2001} and further investigated
in~\cite{gurevich:tcs2005} and~\cite{SchWan10}. This abstraction by partial
updates simplifies the transaction model, as it allows us to model
databases with complex values and to provide an easy-to-explain, yet
still precise model of \emph{multi-level} transactions, where
dependencies of updates of complex database values are dealt with in
terms of compatibility of appropriate value changing operators (see
also~\cite{KiScZh09}). In fact, technically speaking the model we
define here is an ASM refinement (in the sense of~\cite{Boerger02b})
of some of the components of the model published in~\cite{BoeSch14},
namely by a) generalizing the flat transaction model to multi-level
transactions which increase the concurrency in transactions and b)
including an $\ASM{Abort}$ mechanism. Accordingly, the serializability
proof is a refinement of the proof in~\cite{BoeSch14}, as the refined
model is a conservative extension of the model for flat
transactions.\footnote{For a detailed illustration of combined model
  and proof refinement we refer the reader to the Java compiler
  correctness verification in~\cite{BatBoe08}.}

We concentrate on transaction controllers that employ locking
strategies such as the common two-phase locking protocol
(2PL)~\cite{Papadimitriou86}. That is, each transaction first has to
acquire a (read- or write- or more generally operator-) lock for a
shared, possibly nested location, whereby the access to the location to
perform the requested operations is granted. Locks are released after
the transaction has successfully committed and no more access to the
shared locations is necessary. 

There are of course other approaches to transaction handling, see
e.g. \cite{elmasri:2006,gray:1993,KiScZh09,ScRiDr00} and the extensive
literature there covering classical transaction control for flat
transactions, timestamp-based, optimistic and hybrid transaction
control protocols, as well as other non-flat transaction models such
as sagas. To model each of these approaches would fill a book; our
more modest goal here is to concentrate on one typical approach to
illustrate with rigour and in full detail a method by which such
transaction handling techniques can be specified and proved to be
correct. For the same reason we do not consider fairness issues,
though they are important for concurrent runs.

In Section~\ref{sect:MLTransactions} we first give a more detailed
description of the key ideas of multi-level transactions and their
relationship to partial updates.  We define $\ASM{TaCtl}$ and the
operator $TA$ in Section~\ref{sect:TAoperator} and the $\ASM{TaCtl}$
components in Section~\ref{sect:TaCtl}. In Section~\ref{sect:Thm} we
prove the correctness of these definitions. 

We assume the reader to have some basic knowledge of ASMs, covering
the definitions---provided 20 years ago in~\cite{Gurevich94b} and
appearing in textbook form in~\cite[Sect.2.2/4]{BoeSta03}---for what
are ASMs (i.e. their transition rules) and how their execution in
given environments performs state changes by applying sets of updates
to locations. Nevertheless at places where some technical details
about ASMs need to be refered to we briefly describe their notation
and their meaning so that the paper can be understood also by a more
general audience of readers who view ASMs as a semantically
well-founded form of pseudo-code that performs computations over
arbitrary structures.

\section{Multi-Level Transactions and Partial Updates}\label{sect:MLTransactions}

While standard flat transaction models start from a view of operation
sequences at one level, where each operation reads or writes a shared
location---in less abstract terms these are usually records or pages
in secondary storage---the multi-level transaction
model~\cite{BeBeGo89,Weikum86,Weikum91} relaxes this view in various
ways. The key idea is that there are multiple levels, each with its
own view of the data and its own set of operations.

The operations on a higher level may be compatible with one another,
whereas operations on a lower level implementing them are not. As a
motivating example pages in secondary storage and records stored in
these pages can be considered. Updating two different records in the
same page should be compatible, but not simultaneous writing of the whole
page. When updating a particular record, this record should be
locked for writing; as writing the record requires also writing the
page, the page should also be locked. However, the page lock could
immediately be released after writing, as it is sufficient to block
updates to the record until commit. So another transaction could get
access to a different record on the same page with a long lasting lock
on the record and another {\em temporary} lock on the page.

A second key idea of the multi-level transaction model stressed
in~\cite{ScRiDr00,Weikum86,Weikum91} is that some high-level
operations may even be compatible when applied to the same shared
location. Standard examples are addition, subtraction or insertion of
values into a set. For instance, if a field in a record is to be
updated by adding 3 to the stored value, then another operation
subtracting 2 could be executed as well without causing
inconsistencies. Consequently, the strictness of a lock can be
relaxed, as a plus-lock can co-exist with another plus-lock, but must
prevent an arbitrary update or a times-lock (for multiplication).

We will demonstrate in the following sections that these key ideas of
the multi-level transaction model can be easily and precisely captured
by refinement of the ASM-based transaction handler in \cite{BoeSch14}.
Since to execute a step a component ASM $M$ computes a set of updates
(on which the transaction controller $\ASM{TaCtl}$ can speculate for
lock handling etc.), it suffices to incorporate partial updates (as
handled in \cite{SchWan10}) into the model developed in
\cite{BoeSch14}. For the first idea of the multi-level transaction
model we exploit the {\em subsumption} relation between locations
defined in \cite{SchWan10}: a location $l$ subsumes a location
$l^\prime$ iff in all states $S$ the value of $l$, i.e. $eval(l,S)$,
uniquely determines the value of $l^\prime$,
i.e. $eval(l^\prime,S)$. For instance, a value of a page uniquely
determines the values of the records in it, but also a tree value
determines the values of subtrees and leaves. The notion of
subsumption offers a simple realization of the concept of temporary
locks: temporary locks are needed on all subsuming locations.

The second idea of compatible operations can be captured by
introducing particular operation-dependent locks, which fine-tune the
exclusive write locks. Some of these operation-locks may be compatible
with each other, such that different transactions may execute
simultaneously operations on the same location. Naturally, this is
only possible with partial updates defined by an operator $op$ and an
argument $v$. The new value stored at location $l$ is obtained by
evaluating $op(eval(l, S),v)$. If operators are compatible in the
sense that the final result is independent from the order in which
the operators are applied, then several such partial updates can be
executed at the same time.

Thus, the refinement of the concurrent ASM in \cite{BoeSch14} for
handling flat transactions affects several aspects:

\begin{itemize}

\item Each component machine $TA(M,\ASM{TaCtl})$ resulting from the transaction
  operator will have to ask for more specific operation-locks and to
  execute partial updates together with other machines.

\item Each component machine $TA(M,\ASM{TaCtl})$ will also have to
  release temporary locks at the end of each step.

\item In case already the partial updates of $M$ itself are
  incompatible, i.e. are such that they cannot be merged into a single
  genuine update, the machine $TA(M,\ASM{TaCtl})$ should not fire at
  all; instead, it must be completely \ASM{Abort}ed, i.e., all its
  steps will have to be undone immediately.

\item The \ASM{LockHandler} component requires a more sophisticated condition for granting locks, which takes subsumption into account.

\item The \ASM{Recovery} component will have to be extended to capture \ASM{Undo}ing also partial updates, for which inverse operations are required.

\item The \ASM{DeadlockHandler} and \ASM{Commit} components remain unaffected.

\end{itemize}

While these refinements with partial updates to capture multi-level
transactions require only a few changes---which also extend easily to
the serializability proof---they also highlight some not so obvious
deficiencies in the model of multi-level transactions itself. In
\cite{BeBeGo89} it is claimed that each higher-level operation is
implemented by lower-level ones. For instance, an update of a record
requires reading and writing a page. This is also true for
object-oriented or complex value systems. For instance, in
\cite{ScRiDr00} it is anticipated that there could be levels for
objects, records and pages, such that an operation on an object would
require several update operations on records storing parts of the
object. However, in the light of partial updates it is the object that
subsumes the record. This implies that the definition of
level-specific conflict relations \cite{Weikum86,ScRiDr00} with the
condition that a high-level conflict must be reflected in a low-level
one, but not vice versa, is too specific. It is true for fields,
records and pages, but cannot be applied any more, when the
higher-level locations subsume the lower-level ones. On the other
hand, using subsumption for the definition of levels does not work
either, as objects and pages are conceptually different and should not
be considered as residing on the same level. To this end the use of
subsumption between locations makes the idea behind multi-level
transactions much clearer and formally consistent. In particular, the
notion of level itself becomes irrelevant in this setting, so in a
sense the transaction model formalised in this article can be seen as
a moderate generalisation of the multi-level transaction model.

A second strengthening and generalisation of the concept of
multi-level transactions realized in our model comes from the
observation that in order to undo a partial update inverse operations
are not just nice to have, but must exist, because otherwise
recoverability cannot be guaranteed. This also shows that a
transaction model cannot be treated in isolation without taking
recovery into account at the same time.

\section{The Transaction Controller and Operator}
\label{sect:TAoperator}

As explained above, a transaction controller performs the lock
handling, the deadlock detection and handling, the recovery mechanism
(for partial recovery) d the commit or abortion of single
machines---we use Astract State Machines to describe programs. Thus we
define $\ASM{TaCtl}$ as consisting of five components specified in
Sect.~\ref{sect:TaCtl}. We use $\ASM{SmallCaps}$ for rules and
\emph{italics} for functions, sets, predicates, relations.

\begin{asm}
\ASM{TaCtl}=\+
\ASM{LockHandler} \\
\ASM{DeadlockHandler} \\
\ASM{Recovery} \\
\ASM{Commit}\\
\ASM{Abort} 
\end{asm}

\subsection{The Transaction Operator $TA(\bullet,\ASM{TaCtl}$)}

The operator~$TA(\bullet,\ASM{TaCtl})$ transforms the component
machines~$M$ of any concurrent system (in particular an asynchronous,
concurrent ASM \cite{BoeSch15}) ${\cal M} = (M_i)_{i \in I}$ into
components of a concurrent system $TA({\cal M},\ASM{TaCtl})$, where
each component $TA(M_i,\ASM{TaCtl})$ runs as transaction under the
control of~$\ASM{TaCtl}$. Thus $TA({\cal M},\ASM{TaCtl})$ is defined
as follows:\footnote{For notational economy we use the same letters
  $TA$ once to denote an operator applied to a set of component
  machines and $\ASM{TaCtl}$, once to denote an operator applied to
  single component machines and $\ASM{TaCtl}$. From the context it is
  always clear which $TA$ we are talking about.}
\begin{asm}
TA({\cal M},\ASM{TaCtl})= 
                 (TA(M_i,\ASM{TaCtl})_{i \in I},\ASM{TaCtl})
\end{asm}

It remains to expalin the definition of $TA(M,\ASM{TaCtl})$
below. $\ASM{TaCtl}$ keeps a dynamic set $TransAct$ of those
machines~$M$, whose runs it currently has to supervise. This is to
guarantee that $M$ operates in a transactional manner, until it has
$Terminated$ its transactional behavior (so that it can $\ASM{Commit}$
it).\footnote{In this paper we deliberately keep the termination
  criterion abstract so that it can be refined in different ways for
  different transaction instances.}  To turn the behavior of a
machine~$M$ into a transactional one, first of all~$M$ has to register
itself with the controller $\ASM{TaCtl}$, i.e., to be inserted into
the set of $TransAct$ions currently to be handled.  $\ASM{Undo}$ing
some steps~$M$ made in the given transactional run as part of a
recovery, a last-in first-out queue $history(M)$ is needed, which for
each step of~$M$ keeps track of the newly requested locks and of the
recovery updates needed to $\ASM{Restore}$ the values of the
locations~$M$ changed in this step. When~$M$ enters the set
$TransAct$, the $history(M)$ has to be initialized (to the empty
queue).

The crucial transactional feature is that each non-private
(i.e. shared or monitored or
output)\footnote{See~\cite[Ch.2.2.3]{BoeSta03} for the classification
  of locations and functions.} location~$l$ a machine~$M$ needs to
read or write for performing a step has to be $LockedBy(M)$ for this
purpose;~$M$ tries to obtain such locks by calling the
$\ASM{LockHandler}$. In case no $newLocks$ are needed by~$M$ in its
$currState$ or the $\ASM{LockHandler}$ $GrantedLocksTo(M)$, $M$
\emph{canGo} to try to perform its next step: if it cannot fire (due
to an inconsistency of the set $aggregatedUpdSet$\footnote{We borrow
  the name from CoreASM~\cite{CoreASM}.} of updates computed by $M$
from the assignment and the partial update instructions of~$M$, see
below) it calls the $\ASM{Abort}$ component.  If $CanFire(M)$ holds,
we require $\ASM{Fire}(M)$ to perform the $M$-step together with one
step of all $Partner(M)$-machines, i.e. of machines~$N$ that
$CanFire(N)$ simultaneously with~$M$ and share some locations to be
updated with~$M$ (possibly via some compatible update operations on
those locations, see below).\footnote{This view of concurrency is an
  instance of the general definition of concurrent ASMs provided
  in~\cite{BoeSch15}.} This means to $\ASM{Aggregate}$ the (below
called genuine) updates~$M$ yields in its $currState(M)$ together with
the partial updates of~$M$ together with the genuine updates and
partial updates of all $Partner(M)$-machines. In addition a
$\ASM{RecoveryRecord}$ component has to $\ASM{Record}$ for each of
these machines~$N$ in its $history$ the obtained $newLocks$ together
with the $recoveryUpd$ates needed should it become necessary to
$\ASM{Undo}$ the updates contributed by~$N$ to
this~$\ASM{Aggregate}$-step. Then~$M$ continues its transactional
behavior until it is $Terminated$. In case the $\ASM{LockHandler}$
$RefusedLocksTo(M)$, namely because another machine~$N$ in $TransAct$
has some of these locks, $M$ has to $Wait$ for~$N$; in fact it
continues its transactional behavior by calling again the
$\ASM{LockHandler}$ for the needed $newLocks$---until the needed
locked locations are unlocked, when~$N$'s transactional behavior is
$\ASM{Commit}$ed, whereafter a new request for these locks
$GrantedLocksTo(M)$ may become true.\footnote{A refinement (in fact a desirable
  optimization) consists in replacing such a waiting cycle by
  suspending~$M$ until the needed locks are released. Such a
  refinement can be obtained in various ways, a simple one consisting
  in letting~$M$ simply stay in $waitForLocks$ until the $newLocks$
  $CanBeGranted$ and refining $\ASM{LockHandler}$ to only choose pairs
  $(M,L)\in LockRequest$ where it can $\ASM{GrantRequestedLocks}(M,L)$
  and doing nothing otherwise (i.e. defining
  $\ASM{RefuseRequestedLocks}(M,L)=~\SKIP$). See
  Sect.~\ref{sect:TaCtl}.}

As a consequence deadlocks may occur, namely when a cycle occurs in
the transitive closure $Wait^*$ of the $Wait$ relation.  To resolve
such deadlocks the $\ASM{DeadlockHandler}$ component of $\ASM{TaCtl}$
chooses some machines as $Victim$s for a recovery.\footnote{To
  simplify the serializability proof in Sect.\ref{sect:TaCtl} and
  without loss of generality we define a reaction of machines~$M$ to
  their victimization only when they are in $ctl\_state(M)=~$TA-$ctl$
  (not in $ctl\_state(M)=waitForLocks$). This is to guarantee that no
  locks are $Granted$ to a machine as long as it does
  $waitForRecovery$.} After a victimized machine~$M$ is $Recovered$ by
the $\ASM {Recovery}$ component of $\ASM{TaCtl}$ it can exit
its $waitForRecovery$ mode and continue its transactional behavior.
 
This explains the following definition of $TA(M,\ASM{TaCtl})$ as a
control state ASM, i.e. an ASM with a top level Finite State Machine
control structure. We formulate it by the flowchart diagram of
Fig.~\ref{fig:TA(M,C)}, which has a precise control state ASM
semantics (see the definition in~\cite[Ch.2.2.6]{BoeSta03}).\footnote{The
components for the recovery feature are highlighted in the flowchart
by a different colouring.} The
macros which appear in Fig.~\ref{fig:TA(M,C)} are defined in the rest
of this section.

\begin{figure}[htb]
  \begin{center}
  \includegraphics[width=0.9\textwidth]{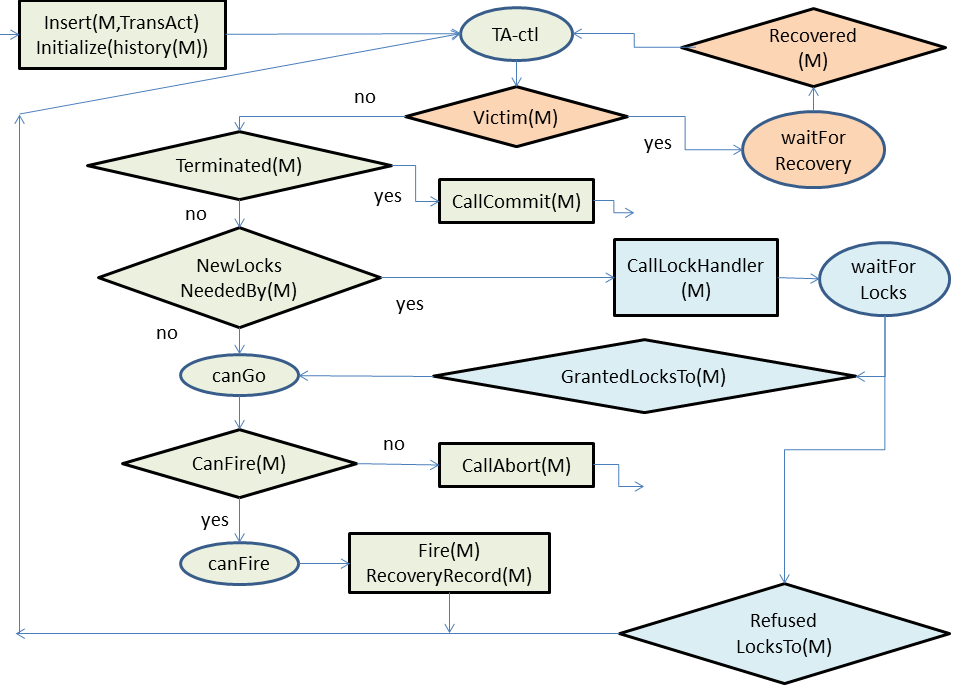}
  \end{center}
  \caption{TA(M,\ASM{TaCtl})}
  \label{fig:TA(M,C)}
\end{figure}

\subsection{The  $TA(M,\ASM{TaCtl})$ Macros}

The predicate $NewLocksNeededBy(M)$ holds, if in the current state
of~$M$ at least one of two cases happens: either $M$ reads some shared
or monitored location, which is not yet $LockedBy(M)$ for reading or
writing, or~$M$ writes some shared or output location which is not yet
$LockedBy(M)$ for the requested write operation. We compute the set of
such needed, but not yet $Locked$ locations by a function $newLocks$
(whose arguments we omit for layout reasons in Fig.\ref{fig:TA(M,C)}).

\begin{asm}
NewLocksNeededBy(M)=\+
   newLocks(M,currState(M)) \not = \emptyset 
\end{asm}

Whether a lock for a location can be granted to a machine depends on
the kind of read or write operation the machine wants to perform on
the location. 

\subsubsection{Updates and partial updates.} 

In basic ASMs a write operation is denoted by \emph{assignment
  instructions} of form $s:=t$ 
resulting for $s=f(t_1,\ldots,t_n)$ in any given state~$S$ in an
update of the location $l=(f,(eval(t_1,S),\ldots,eval(t_n,S))$ by the
value $eval(t,S)$~\cite[pg.29]{BoeSta03}. Here $eval(t,S)$ denotes the
evaluation of~$t$ in state~$S$ (under a given interpretation~$I$ of
free variables). We call such updates $(l,val)$ \emph{genuine}
(in~\cite{SchWan10} they are called exclusive) to distinguish them
from partial updates. The reader who is not knowledgeable about ASMs
may interpret locations $(f,args)$ correctly as array variables with
variable name~$f$ and index $args$.

Analogously, we denote write operations that involve partial updates via an
operation $op$ by \emph{update instructions} of form
\[s:=_{op}t\]
which require an overall update of the location
$(f,(eval(t_1,S),\ldots,eval(t_n,S))$ by a value to which $s:=_{op}t$
contributes by the value $op(eval(s,S),eval(t,S))$. A typical example
of such operations is parallel counting (used already
in~\cite{BoDuRo94}) where say seven occurences of a partial update instruction
\[x:=_{parCount}x+1\]
in a state~$S$ are aggregated into a genuine update
$(x,eval(x+7,S))$. Other examples are tree manipulation by
simultaneous updates of independent subtrees or more generally term
rewriting by simultaneous independent subterm updates, etc.,
see~\cite{ScRiDr00,SchWan10}. $\ASM{Aggregate}$ (which is implemented
as a component in CoreASM~\cite{CoreASM}) specifies how to compute and
perform the desired overall update effect, i.e. the genuine update set
yielded by the set of all genuine and the multiset of all partial
updates involving any location~$l$ and all other higher or lower level
location updates the new value of~$l$ may depend upon due to an update
to be performed at that level by some machine in the considered step.

Therefore, a location can be $LockedBy(M)$ for reading
($Locked(l,M,Read)$) or for writing ($Locked(l,M,Write)$) via a genuine
update or for a partial update using operation~$op$
($Locked(l,M,op)$). We also use $Locked(l,M,temp)$ for a temporary
lock of a location $l$. Same as a genuine write-lock such a temporary
lock blocks location $l$ for exclusive use by $M$. However, such
temporary locks will be immediately released at the end of a single
step of $M$. As explained in Section \ref{sect:MLTransactions} the
purpose of such temporary locks is to ensure that an implied write
operation on a subsuming location (i.e., a partial update) can be
executed, but the lock is not required any more after completion of
the step, as other non-conflicting partial updates should not be
prohibited.

Even if $Locked(l,M,op)$ temporarily (case $op=temp$) or for a partial
update operation (case $op \neq Read,Write$) machine~$M$ still needs a
lock to be allowed to Read~$l$ because for a partial update location a
different machine could acquire another compatible operation lock
on~$l$ that is not controllable by~$M$ alone. For this reason partial
update operations are defined below to be incompatible with Read and
genuine Write.

\begin{asm}
newLocks(M,currState(M))\footnote{By the second argument 
  $currState(M)$ of $newLocks$ we indicate 
  that this function of~$M$ is a dynamic function which is evaluated in each 
  state of~$M$, namely by computing in this state the sets $ReadLoc(M)$ and 
  $WriteLoc(M)$; see Sect.~\ref{sect:Thm} for the detailed definition.}
  =\+
     \{(l,Read) \mid l \in R\mbox{-}Loc  \AND \+
         \NOT Locked(l,M,Read) \AND ~\NOT Locked(l,M,Write) \}\-
     \cup ~\{ (l,o) \mid l \in W\mbox{-}Loc 
          \AND o \in \{ Write \} \cup Opn \AND ~\NOT Locked(l,M,o) \} \\
    \cup ~\{ (l,temp) \mid \exists l^\prime \in W\mbox{-}Loc \;  l \neq l^\prime
          \AND l \;\text{subsumes}\; l^\prime \} \-
\WHERE \\
R\mbox{-}Loc = 
            ReadLoc(M,currState(M)) \cap (SharedLoc(M) \cup MonitoredLoc(M))\\
W\mbox{-}Loc = 
           WriteLoc(M,currState(M)) \cap (SharedLoc(M) \cup OutputLoc(M))\-
LockedBy(M)=  \+
   \{l \mid Locked(l,M,Read) \OR Locked(l,M,Write) \OR Locked(l,M,temp) \OR \+
                 Locked(l,M,op) \FORSOME op \in Opn \}
\end{asm}

To $\ASM{CallLockHandler}$ for the $newLocks$ requested by~$M$ in its
$currState(M)$ means to $\ASM{Insert}(M)$ into the
$\ASM{LockHandler}$'s set of to be handled $LockRequest$s. Similarly
we let $\ASM{CallCommit(M)}$ resp.  $\ASM{CallAbort(M)}$ stand for
insertion of~$M$ into a set $CommitRequest$ resp. $AbortRequest$ of
the $\ASM{Commit}$ resp. $\ASM{Abort}$ component.

\begin{asm}
\ASM{CallLockHandler}(M)=~\ASM{Insert}(M,LockRequest)\\
\ASM{CallCommit}(M)=~\ASM{Insert}(M,CommitRequest)\\
\ASM{CallAbort}(M)=~\ASM{Insert}(M,AbortRequest)
\end{asm}

Once a machine $canGo$ because it has acquired all needed locks for
its next proper step, it must be checked whether the
$aggregatedUpdSet(M,currState(M))$ it yields in its current state is
consistent so that $CanFire(M)$: if this is not the case, $M$ is
$\ASM{Abort}$ed whereby it interrupts its transactional
behavior. 

\begin{asm}
CanFire(M)=  Consistent(aggregatedUpdSet(M,currState(M))).
\end{asm}

Here $aggregatedUpdSet(M,S)$ is defined as the set of updates~$M$
yields\footnote{See the definition in~\cite[Table 2.2
    pg.74]{BoeSta03}.}  in state~$S$, once the resulting genuine
updates have been computed for all partial updates to be performed
by~$M$ in~$S$.\footnote{In CoreASM~\cite{CoreASM} this computation is
  done by corresponding plug-ins.} If this update set is consistent,
to $\ASM{Fire}(M)$ $\ASM{Aggregate}$ performs not only the (genuine
and partial) updates of~$M$, but also those of any other
$Partner(M)$-machine~$N$ which shares some to-be-updated location
with~$M$ and $CanFire(N)$ simultaneously with~$M$.

\begin{asm}
\ASM{Fire}(M)=\+
    \FORALL N \in Partner(M) \DO  N \\
    \ASM{ReleaseTempLocks}(M) \-
\WHERE \+
Partner(M)=~
   \{ N \mid ShareUpdLocs(M,N) \AND mode(N)=canFire \} \\
\ASM{ReleaseTempLocks}(M) = \+
     \FORALL l ~~ Locked (l,M,temp) := false
\end{asm}

The constraints defined in the next section for $GrantedLocksTo(M)$ and the
consistency condition for $aggregatedUpdSet$s guarantee that 
$\ASM{Fire}(M)$ computes and performs a consistent update set.

\subsubsection{Remark on notation.} 

As usual with programming languages, for
ASMs we consider (names for) functions, rules, locations, etc., as
elements of the universe for which sets (like $ReadLoc$, $WriteLocs$)
and relations (like subsumption) can be mathematically defined and
used in rules. In accordance with usual linguistic reflection notation
we also quantify over such elements, e.g. in $\FORALL N \in SetOfAsm
\DO N$, meaning that $\DO N$ stands for an execution of (a step of)
the ASM denoted by the logical variable $N$.

The $\ASM{RecoveryRecord}(M)$ component has to $\ASM{Record}$ for each
$Partner(M)$-machine its $recoveryUpd$ates (defined below where we
need the details for the $\ASM{Recovery}$ machine) and the obtained
$newLocks$.

\begin{asm}
\ASM{RecoveryRecord}(M)=~\FORALL N \in Partner(M)\+
   \ASM{Record}(recoveryUpd(N,currState(N)),\+
           newLocks(N,currState(N)),N)\dec\-
\ASM{Record}(recUpdSet,lockSet,N) = \+
             \ASM{Append}((recUpdSet,lockSet),history(N))
\end{asm}

\subsubsection{Remark on nondeterminism.} 

The ASM framework provides two ways to
deal with nondeterminism. `True' nondeterminism can be expressed using
the $\CHOOSE$ construct to define machines of form

\begin{asm}
M =  \CHOOSE x \WITH \alpha \IN r(x)
\end{asm}

\noindent where $r$ has to be an ASM rule. Nondeterminism can also be
modeled `from outside' by using choice functions, say $select$, in
machines of form
\begin{asm}
N =  r(select(\alpha))
\end{asm}

\noindent where in the view of the transition rules everything is
deterministic once a definition of the choice function is
given. Using one or the other form of nondeterminism influences the
underlying logic for ASMs (see \cite[Ch.8.1]{BoeSta03}).

The locks acquired for a machine $M$ as above depend on the chosen value
$a$ for $x$ so that when $M$ performs its next step it must have the
same value $a$ for $x$ to execute $r(x)$. To `synchronize' this
choice of $a$ for $x$ for lock acquisition and rule execution we
assume here nondeterminism in component machines $M \in TransAct$ to be
expressed by choice functions.

\section{The Transaction Controller Components}
\label{sect:TaCtl}

\subsection{The $\ASM{Commit}$ Component}

A $\ASM{CallCommit(M)}$ by machine~$M$ enables the $\ASM{Commit}$
component, which handles one at a time $CommitRequest$s. For this we
use the $\CHOOSE$ operator, so we can leave the order in which the
$CommitRequest$s are handled refinable by different instantiations of
$\ASM{TaCtl}$.

$\ASM{Commit}$ing~$M$ means to $\ASM{Unlock}$ all locations~$l$ that
are $LockedBy(M)$.\footnote{$\ASM{Unlock}$ is called only in states
  where~$M$ has no $temp$orary lock.} Note that each lock obtained
by~$M$ remains with~$M$ until the end of~$M$'s transactional
behavior. Since~$M$ performs a $\ASM{CallCommit(M)}$ when it has
$Terminated$ its transactional computation, nothing more has to be
done to $\ASM{Commit}$ $M$ besides deleting~$M$ from the sets of
$CommitRequest$s and still to be handled $TransAct$ions.\footnote{We
  omit clearing the $history(M)$ queue since it is initialized
  when~$M$ is inserted into $TransAct(\ASM{TaCtl})$.}

\begin{asm}
\ASM{Commit} =\+
  \IF CommitRequest \not = \emptyset \THEN \+
     \CHOOSE M \in CommitRequest ~\ASM{Commit}(M) \-
  \WHERE \+
    \ASM{Commit}(M) = \+
       \FORALL l \in LockedBy(M) ~~ \ASM{Unlock}(l,M) \\
       \ASM{Delete}(M,CommitRequest)\\
       \ASM{Delete}(M,TransAct)\-
    \ASM{Unlock}(l,M)= ~\FORALL o \in \{Read,Write\} \cup Opn \+
                        Locked(l,M,o):=false 
\end{asm}

The locations $(Locked,(l,M,o))$ are shared by the $\ASM{Commit}$,
$\ASM{LockHandler}$ and $\ASM{Recovery}$ components, but these
components never have the same~$M$ simultaneously in their request or
$Victim$ set, respectively: When ~$M$ has performed a
$\ASM{CallCommit(M)}$, it has $Terminated$ its transactional
computation and does not participate any more in any $LockRequest$ or
$Victim$ization. Furthermore, by definition no~$M$ can at the same
time issue a $LockRequest$ (possibly triggering the
$\ASM{LockHandler}$ component) and be a $Victim$ (possibly triggering
the $\ASM{Recovery}$ component).

\subsection{The $\ASM{LockHandler}$ Component}

As for $\ASM{Commit}$ we use the $\CHOOSE$ operator also for the
$\ASM{LockHandler}$ to leave the order in which the $LockRequest$s are
handled refinable by different instantiations of $\ASM{TaCtl}$.

The strategy we adopted in~\cite{BoeSch14} for lock handling with only
genuine updates was to refuse all locks for locations requested by~$M$,
if at least one of the following two cases occurs:

\begin{itemize}
\item some of the requested locations is $Locked$ by another
  transactional machine~$N \in TransAct$ for writing,
\item some of the requested locations is a $WriteLoc$ation in $W\mbox{-}Loc$ that is
  $Locked$ by another transactional machine~$N \in TransAct$ for reading.
\end{itemize}

In other words, read operations of different machines are compatible
and upgrades from read to write locks are possible. In the presence of
partial updates, which have to be simultaneously performed by one or
more transactional machines this compatibility relation has to be
extended to partial operations, but guaranteeing the consistency of
the result of the $\ASM{Aggregate}$ mechanism which binds together
shared updates to a same location. We adopt the following constraints
defined in~\cite{SchWan10}:

\begin{itemize}
\item A genuine Write is incompatible with a Read or genuine Write or
  any partial operation $op \in Opn$ of any other machine.
\item A Read is incompatible with any Write (whether genuine or
  involving a partial operation $op \in Opn$).
\item Two partial operations $op,op' \in Opn$ are incompatible on a
  location~$l$ if in some state applying to update~$l$ first $op$ then
  $op'$ yields a different result from first applying $op'$ then $op$.
\end{itemize} 

However, to guarantee the serializability of transactions in the
presence of partial updates of complex data structures consistency is
needed also in case one update concerns a substructure of another
update. Therefore we stipulate that a lock request for~$l$
$CannotBeGranted$ to a machine~$M$ as long as a location~$l'$ which
subsumes~$l$ is $Locked$ by another machine~$N$. The subsumption definition
is taken from~\cite[Def.2.1]{SchWan10}:

\begin{asm}
l' \mbox{ subsumes } l = ~  \FORALL S ~ 
   eval(l',S) \mbox{ uniquely determines } eval(l,S)
\end{asm}

To $\ASM{RefuseRequestedLocks}$ it suffices to set the communication
interface $RefusedLocksTo(M)$ of $TA(M,\ASM{TaCtl})$; this makes~$M$ $Wait$ for
each location~$l$ and operation~$o$ for which the lock
$CannotBeGranted$ to~$M$.

\begin{asm}
\ASM{LockHandler} =\+
\IF LockRequest \not = \emptyset \THEN \+
   \CHOOSE M \in LockRequest ~  \ASM{HandleLockRequest}(M)\-
\WHERE \+
 \ASM{HandleLockRequest}(M) = \+
    \LET L = newLocks(M,currState(M))\+
      \IF CannotBeGranted(M,L) \+
         \THEN  ~ \ASM{RefuseRequestedLocks}(M,L)\\
         \ELSE  ~ \ASM{GrantRequestedLocks}(M,L) \-
      \ASM{Delete}((M,L),LockRequest)\dec\-         
 CannotBeGranted(M,L)=\+
    \FORSOME (l,o) \in L ~ CannotBeGranted(l,M,o)\-
 CannotBeGranted(l,M,o)=\+
       \FORSOME N \in TransAct \setminus \{M\} ~~ \mbox{ Blocks } (N,l,o)\- 
 \mbox{ Blocks } (N,l,o)=~ \FORSOME o'\+
   Locked(l,N,o') \AND ~ \NOT Compatible(o,o',l) \+
      \OR ~ \FORSOME l'~ Locked(l',N,o') \AND l' \mbox{ subsumes } l \dec\-
 \ASM{RefuseRequestedLocks}(M,L) =  (RefusedocksTo(M):=true)\\
 \ASM{GrantRequestedLocks}(M,L)=\+
    \FORALL (l,o) \in L ~~ Locked(l,M,o):=true\\ 
    GrantedLocksTo(M):=true   
\end{asm}

\subsection{The $\ASM{DeadlockHandler}$ Component}

A $Deadlock$ originates if two machines are in a $Wait$ cycle, i.e.,
they wait for each other. In other words, a deadlock occurs, when for
some (not yet $Victim$ized) machine~$M$ the pair $(M,M)$ is in the
transitive (not reflexive) closure $Wait^*$ of $Wait$. In this case
the $\ASM{DeadlockHandler}$ selects for recovery a (typically minimal)
subset of $Deadlocked$ transactions $toResolve$---they are
$Victim$ized to $waitForRecovery$, in which mode (control state) they
are backtracked until they become $Recovered$. The selection criteria
are intrinsically specific for particular transaction controllers,
driving a usually rather complex selection algorithm in terms of
number of conflict partners, priorities, waiting time, etc.  In this
paper we leave their specification for $\ASM{TaCtl}$ abstract (read:
refinable in different directions) by using the $\CHOOSE$ operator.

\begin{asm}
\ASM{DeadlockHandler} =\+
\IF Deadlocked \cap \overline{Victim} \not = \emptyset \THEN 
                  \mbox{ // there is a Wait cycle}\+
    \CHOOSE toResolve \subseteq Deadlocked \cap \overline{Victim} \+  
      \FORALL M \in toResolve~ Victim(M):=true\dec\-  
\WHERE \+
Deadlocked = \{M \mid (M,M) \in Wait^*\}\\
Wait^*= \mbox{ TransitiveClosure}(Wait)\\
Wait(M,N) = ~ \FORSOME (l,o) ~ Wait(M,(l,o),N) \\
Wait(M,(l,o),N) = \+
     (l,o) \in newLocks(M,currState(M)) \AND   N \in TransAct \setminus \{M\} \+ 
            \AND   ~ \mbox{ Blocks } (N,l,o)
\end{asm}

\subsection{The $\ASM{Recovery}$ Component}
\label{sect:recovery}

Also for the $\ASM{Recovery}$ component we use the $\CHOOSE$ operator
to leave the order in which the $Victim$s are chosen for recovery
refinable by different instantiations of $\ASM{TaCtl}$. In order to be
$Recovered$ a machine~$M$ is backtracked by $\ASM{Undo}(M)$ steps
until $M$ is not $Deadlocked$ any more, in which case it is
deleted from the set of $Victim$s, so that by definition it is
$Recovered$. This happens at the latest when $history(M)$ has become
empty.

\begin{asm}
\ASM{Recovery} =\+
   \IF Victim \not = \emptyset \THEN \+
      \CHOOSE M \in Victim ~\ASM{TryToRecover}(M)\dec \-          
\WHERE \+
\ASM{TryToRecover}(M) = \+
     \IF M \not \in Deadlocked \THEN Victim(M):=false \+
                \ELSE ~ \ASM{Undo}(M)\dec\-
Recovered = \+
   \{M \mid ctl\mbox{-}state(M)=waitForRecovery \AND M \not \in Victim\}
\end{asm}

To define an $\ASM{Undo}(M)$ step we have to provide the details of
the function $recoveryUpd$ used above in $\ASM{RecoveryRecord}$. This
function collects for any given machine~$M$ and state~$S$ first of all
the $genuineRecoveryUpd$ates by which one can $\ASM{Restore}$ the
overwritten values in $GenuineWriteLoc$ations (i.e. locations to
which~$M$ in $S$ writes via an assignment instruction); in~\cite{BoeSch14}
where we considered only genuine updates this function was called
$overWrittenVal$. 

In addition, for each to be $\ASM{Aggregate}$d
update instruction $s:=_{op}t \in UpdInstr(M,S)$ $recoveryUpd$ collects the information to
compute the `inverse' update for~$M$, information that is needed when
the controller has to $\ASM{Undo}$ at the concerned location the
effect of that partial update by~$M$ (but not of simultaneous partial
updates concerning the same location by other machines). This
information consists in an operation $op'$ with the appropriate
value~$v'$ for its second argument, whereas the first argument is
provided only when the $\ASM{Undo}$ takes place. For the approach to
ASMs with partial updates defined in~\cite{SchWan10} and adopted here
one has to postulate that such operations and values $(op',v')$ which
are $inverse$ to partial update operations $(op,v)$ (where
$v=eval(t,currState(M))$) are defined and satisfy the following
constraint for partial update instructions:

\begin{asm}
\mbox{\bf Inverse Operation Postulate}\+
     \FORALL s:=_{op}t \in UpdInstr(M) \FORALL (op,v) \THEREIS (op',v') \WITH \+
          \FORALL w~~op'(op(w,v),v')=w
\end{asm}

This postulate can be justified by the requirement that any
transaction should be {\em recoverable}~\cite{ScRiDr00}. If
recoverability cannot be guaranteed, a transaction controller must (in
principle) be able to undo updates that were issued long ago, which
would be completely infeasible for any real system where once a
transaction has committed, it can leave the system, and none of its
updates can be undone any more. As partial updates operations
$(op_1,v_1)$ and $(op_2,v_2)$ from two different machines $M$ and $N$
could be executed simultaneously, for each of these operations it must
be forseen that it may be undone, even if the issueing transaction for the
other operation has already committed---this situation has become
possible. As the original value at location $l$ at the time of the
partial update by $M$ using $(op,v)$ is no longer available---anyway,
it may have been updated many times by other compatible partial
updates---$M$ must be able to undo its part of the update independently
from all other updates including \ASM{Undo}ne ones, i.e. to say after
\ASM{Undo}ing $(op,v)$, the resulting values at location $l$ must be
just the one that would have resulted, if only all other (not yet
\ASM{Undo}ne) partial updates had been executed. This is guaranteed by
the inverse operation postulate.

In the original work on multi-level transactions including \cite{BeBeGo89}
recovery is not handled at all, which leads to misleading conclusions
that commutativity of high-level operations---those that can be
defined by partial updates---is sufficient for obtaining increased
transaction throughput by means of additional permitted
schedules. However, commutativity (better called {\em operator
  compatibility}, see \cite{SchWan10}) has to be complemented by the
inverse operation postulate to ensure recoverability. Inverse operators
are claimed in the MLR recovery system \cite{Lomet92}, but no
satisfactory justification was given.

There may be more than one update instruction~$M$ performs for the
same location so that the corresponding inverse
$partialRecoveryUpd$ates have to be $\ASM{Aggregate}$d with the
$genuineRecoveryUpd$ates by $\ASM{Restore}$.

\begin{asm}
recoveryUpd(M,S) = \+
    (genuineRecoveryUpd(M,S),partialRecoveryUpd(M,S))\-
genuineRecoveryUpd(M,S)= ~ \{((f,args),val) \mid \+
  (f,args) \in GenuineWriteLoc(M,S) \AND val = eval(f(args),S)\} \-
partialRecoveryUpd(M,S)=~ \{(l,(op',v')) \mid \+
   \FORSOME f(t_1,\ldots,t_n):=_{op}t \in UpdInstr(M,S)\+
       l=(f,(eval(t_1,S),\ldots,eval(t_n,S))) \AND  \+
       (op',v')=inverse(op,eval(t,S))\}  
\end{asm}

\begin{asm}
\ASM{Undo}(M)= \+
  \LET (Upds,Locks) = youngest(history(M))\+
     \ASM{Restore}(Upds,M)\\
     \ASM{Release}(Locks,M)\\
     \ASM{Delete}((Upds,Locks),history(M))\dec\-
\WHERE \+
   \ASM{Restore}((G,P),M) =~ \ASM{Aggregate}(G \cup 
                      \mbox{ // NB. P is a multiset}\+
          \{\mid((f,args),op'(eval(f(args),currState(M)),v') \mid \+
                  ((f,args),(op',v')) \in P \mid\})
                   \mbox{  // NB. multiset notation } \{\mid ~ \mid\}\dec\-
   \ASM{Release}(L,M)= ~ 
        \FORALL l \in L ~ \ASM{Unlock}(l,M)
\end{asm}

The inverse operation postulate cannot guarantee that the
inverse operations commute with each other in general. However, it can
be guaranteed that on the values, to which the inverse operations are
applied in \ASM{Undo} steps, commutativity holds: For this let
$(op_i^\prime, v_i^\prime)$ be inverse for $(op_i, v_i)$ for $i=1,2$,
such that both operations $op_i$ are compatible and both inverse
operations $op_i^\prime$ have to execute simultaneously on location
$l$. That is, if $v$ is the actual value of $l$ in the current state,
we need to show $op_1^\prime(op_2^\prime(v,v_2^\prime),v_1^\prime) =
op_2^\prime(op_1^\prime(v,v_1^\prime),v_2^\prime)$. As these are
\ASM{Undo} operations, we can assume that $(op_i, v_i)$ for $i=1,2$
have been executed on some previous value of location $l$. Thus, due
to commutativity we must have $v= op_1(op_2(v^\prime,v_2),v_1) =
op_2(op_1(v^\prime,v_1),v_2)$ for some value $v^\prime$. From this we
get
\begin{gather*}
op_1^\prime(op_2^\prime(v,v_2^\prime),v_1^\prime) = 
op_1^\prime(op_2^\prime(op_2(op_1(v^\prime,v_1),v_2),v_2^\prime),v_1^\prime) = \\
op_1^\prime(op_1(v^\prime,v_1),v_1^\prime) = 
v^\prime =
op_2^\prime(op_2(v^\prime,v_2),v_2^\prime) = \\
op_2^\prime(op_1^\prime(op_1(op_2(v^\prime,v_2),v_1),v_1^\prime),v_2^\prime) =
op_2^\prime(op_1^\prime(v,v_1^\prime),v_2^\prime)
\end{gather*}

Note that in our description of the \ASM{DeadlockHandler} and the
(partial) \ASM{Recovery} we deliberately left the strategy for victim
selection and $\ASM{Undo}$ abstract, so fairness considerations will have to be
discussed elsewhere. It is clear that if always the same victim is
selected for partial recovery, the same deadlocks may be created again
and again. However, it is well known that fairness can be achieved by
choosing an appropriate victim selection strategy.

\subsection{The $\ASM{Abort}$ Component}

The $\ASM{Abort}$ component can be succinctly defined as turbo
ASM~\cite[Ch.4.1]{BoeSta03}:

\begin{asm}
\ASM{Abort}= ~ \FORALL M \in AbortRequest\+
   \ITERATE ~ \ASM{Undo}(M) \UNTIL history(M)=\emptyset\\
   Delete(M,TransAct)
\end{asm}
We use the $\ITERATE$ construct only here and do this for notational
convenience to avoid tedious programming of an iteration. We do not use 
$\ITERATE$ to form component ASMs which go into $TransAct$.
 
\section{Correctness Theorem}
\label{sect:Thm}

In this section we show the desired correctness property: if all
monitored or shared locations of any $M_i$ are output or controlled
locations of some other $M_j$ and all output locations of any $M_i$
are monitored or shared locations of some other $M_j$ (closed system
assumption)\footnote{This assumption means that the environment is
  assumed to be one of the component machines.}, each run
of $TA({\cal M},\ASM{TaCtl})$ is equivalent to a serialization of the
terminating $M_i$-runs, namely the $M_{i_{1}}$-run followed by the
$M_{i_{2}}$-run etc., where $M_{i_{j}}$ is the $j$-th machine of $\cal
M$ which performs a commit in the $TA({\cal M},\ASM{TaCtl})$ run. To
simplify the exposition (i.e. the formulation of statement and proof
of the theorem) we only consider machine steps which take place under
the transaction control, in other words we abstract from any
step~$M_i$ makes before being $\ASM{Insert}$ed into or after being
$\ASM{Delet}$ed from the set $TransAct$ of machines which currently
run under the control of $\ASM{TaCtl}$.

First of all we have to make precise what a {\em serial} multi-agent
ASM run is and what {\em equivalence} of $TA({\cal M},\ASM{TaCtl})$
runs means in the general multi-agent ASM framework.

\subsection{Definition of run equivalence} 

Let $S_0, S_1, S_2, \dots$ be a (finite or infinite) run of $TA({\cal
  M},\ASM{TaCtl})$. In general we may assume that \ASM{TaCtl} runs
forever, whereas each machine $M \in \mathcal{M}$ running as
transaction will be $Terminated$ or $\ASM{Abort}$ed at some time --
once $\ASM{Commit}$ed $M$ will only change values of non-shared and
non-output locations\footnote{It is possible that one ASM $M$ enters
  several times as a transaction controlled by \ASM{TaCtl}. However,
  in this case each of these registrations will be counted as a
  separate transaction, i.e. as different ASMs in $\mathcal{M}$.}. To
simplify the proof but without loss of generality we assume that each
update concerning an $\ASM{Abort}$ed machine is eliminated from the
run. For $i = 0,1,2,\dots$ let $\Delta_i,\Gamma_i$ denote the unique
set of genuine updates resp. multiset of partial updates leading to an
$\ASM{Aggregate}$d consistent set of updates defining the transition
from $S_i$ to $S_{i+1}$. By definition of $TA({\cal M},\ASM{TaCtl})$
each $\Delta_i,\Gamma_i$ is the union of the corresponding sets
resp. multisets\footnote{We indicate multiset operations by an upper
  index~$+$} of the agents executing $M \in \mathcal{M}$
resp. $\ASM{TaCtl}$:
\[ \Delta_i = \bigcup\limits_{M \in \mathcal{M}} \Delta_i(M) 
      \cup \Delta_i(\ASM{TaCtl})~~~
    \Gamma_i = \bigcup^+\limits_{M \in \mathcal{M}} \Gamma_i(M) 
      \cup^+ \Gamma_i(\ASM{TaCtl}). \]

$\Delta_i(M)$ contains the genuine and $\Gamma_i(M)$ the
partial updates defined by the machine $TA(M,\ASM{TaCtl})$ in state
$S_i$\footnote{We use the shorthand notation $\Delta_i(M)$ to denote
  $\Delta_i(TA(M,\ASM{TaCtl}))$, analogously $\Gamma_i(M)$; in other
  words we speak about steps and updates of~$M$ also when they really
  are done by~$TA(M,\ASM{TaCtl})$. Mainly this is about transitions
  between the control states, namely TA-$ctl$, $waitForLocks$,
  $waitForRecovery$ (see Fig.\ref{fig:TA(M,C)}), which are performed
  during the run of~$M$ under the control of the transaction
  controller $\ASM{TaCtl}$. When we want to name an original update
  of~$M$ (not one of the updates of $ctl\_state(M)$ or of the
  $\ASM{Record}$ component) we call it a proper $M$-update.},
$\Delta_i(\ASM{TaCtl})$ resp. $\Gamma_i(\ASM{TaCtl})$ contain the
genuine resp. partial updates by the transaction controller in this
state. The sequence
\[\Delta_0(M),\Gamma_0(M),\Delta_1(M),\Gamma_1(M), \Delta_2(M),\Gamma_2(M) \dots\]
will be called the {\em schedule} of $M$ (for the given transactional
run).

To generalise for transactional ASM runs the equivalence of
transaction schedules known from database systems
\cite[p.621ff.]{elmasri:2006} we now define two {\em cleansing
  operations} for ASM schedules. By the first one (i) we eliminate all
(in particular unsuccessful-lock-request) computation segments which
are without proper $M$-updates; by the second one (ii) we eliminate
all $M$-steps which are related to a later $\ASM{Undo}(M)$ step by the
$\ASM{Recovery}$ component:

\begin{enumerate}

\item Delete from the schedule of $M$ each $\Delta_i(M),\Gamma_i(M)$ where one
  of the following two properties holds:

    \begin{itemize}
    \item $\Delta_i(M)=\Gamma_i(M)=\emptyset$ ($M$ contributes no
      update to $S_i$),
   \item $\Delta_i(M)$ belongs to a step of an $M$-computation segment
     where~$M$ in its $ctl\_state(M)=$TA-$ctl$ does
     $\ASM{CallLockHandler}(M)$ and in its next step moves
     from $waitForLocks$ back to control state TA$-ctl$ because the
     $\ASM{LockHandler}$ $RefusedLocksTo(M)$.\footnote{By
       eliminating this $\ASM{CallLockHandler}(M)$ step also the
       corresponding $\ASM{LockHandler}$ step
       $\ASM{HandleLockRequest}(M)$ disappears in the run.}
     \end{itemize}

In such computation steps~$M$ makes no proper update.

\item Repeat choosing from the schedule of $M$ a pair
  $\Delta_j(M),\Gamma_j(M)$ with later $\Delta_{j'}(M),\Gamma_{j'}(M)$
  ($j<j'$) which belong to consecutive
  $M$-Recovery entry resp. exit steps defined as follows:

  \begin{itemize}
    \item a (say $M$-RecoveryEntry) step whereby~$M$ in state $S_j$
      moves from TA-$ctl$ to $waitForRecovery$ because it became a
      $Victim$,
    \item the next $M$-step (say $M$-RecoveryExit) whereby~$M$ in
      state $S_{j'}$ moves back to control state TA-$ctl$ because it has
      been $Recovered$.
  \end{itemize} 

In these two $M$-Recovery steps~$M$ makes no proper update. Delete:

\begin{enumerate}
\item $\Delta_j(M),\Gamma_j(M)$ and $\Delta_{j'}(M),\Gamma_{j'}(M)$, 
\item the $((Victim,M),true)$ update from the corresponding
  $\Delta_t(\ASM{TaCtl})$ ($t< j$) which in state $S_j$ triggered the
  $M$-RecoveryEntry,
\item $\ASM{TryToRecover}(M)$-updates in any
  $\Delta_{i+k}(\ASM{TaCtl}),\Gamma_{i+k}(\ASM{TaCtl})$ between
  the considered $M$-RecoveryEntry and $M$-RecoveryExit step
  (for~$i$ as below with $i<j<i+k<j'$),
\item each $\Delta_{i'}(M),\Gamma_{i'}(M)$ belonging to the
  $M$-computation segment from TA-$ctl$ back to TA-$ctl$ which
  contains the proper $M$-step in $S_i$ that is
  $\ASM{UNDO}$ne in $S_{i+k}$ by the considered
  $\ASM{TryToRecover}(M)$ step. Besides control state and
  $\ASM{Record}$ updates these $\Delta_{i'}(M)$ contain genuine
  updates $(\ell,v)$ with $\ell =
  (f,(eval(t_1,S_i),\dots,eval(t_n,S_i)))$ where the corresponding
  $\ASM{Undo}$ updates are 
\[(\ell,eval(f(t_1,\dots,t_n),S_i)) \in  \Delta_{i+k}(\ASM{TaCtl})\]
 
the $\Gamma_{i'}(M)$ contain partial updates
\[(f,(eval(t_1,S),\ldots , eval(t_n,S))),op(eval(f(t_1,\ldots,t_n),S)),eval(t),S))\]

for update instructions $f(t_1,\ldots,t_n):=_{op}t$ of~$M$ in $S_{i'}$
whose effect is $\ASM{Undo}$ne when $\ASM{Recovery}$
$\ASM{Aggregate}$s the $M$-specific partial update $(l,(op',v'))$ with
the inverse operation $(op',v')$ to $(op,eval(t,S_{i'}))$ on~$l$.

 \item the $\ASM{HandleLockRequest}(M)$-updates in
   $\Delta_{l\prime}(\ASM{TaCtl})$ corresponding to $M$'s step by
   $\ASM{CallLockHandler}$ (if any: in case $newLocks$ are needed
   for the proper $M$-step in $S_i$) in state $S_l$ ($l<l^\prime<i$).

 \end{enumerate}

\end{enumerate}

The sequence $\Delta_{i_1}(M), \Gamma_{i_1}(M),
\Delta_{i_2}(M),\Gamma_{i_1}(M), \dots$ with $i_1 < i_2 < \dots$
resulting from the application of the two cleansing operations as long
as possible will be called the {\em cleansed schedule} of $M$ (for the
given run).  Note that the sequence is uniquely defined because
confluence results from the fact that the deletion order chosen in
step (i) or step (ii) does not matter.

Before defining the equivalence of transactional ASM runs let us
remark that $TA({\cal M},\ASM{TaCtl})$ has indeed several runs, even
for the same initial state $S_0$. This is due to the fact that a lot
of non-determinism is involved in the definition of this ASM. First,
the submachines of \ASM{TaCtl} are non-deterministic:

\begin{itemize}

\item In case several machines $M, M^{\prime} \in \mathcal{M}$ request
  conflicting locks at the same time, the \ASM{LockHandler} can only
  grant the requested locks for one of these machines.

\item Commit requests are executed in random order by the \ASM{Commit} submachine.

\item The submachine \ASM{DeadlockHandler} chooses a set of victims,
  and this selection has been deliberately left abstract.

\item The \ASM{Recovery} submachine chooses in each step a victim $M$,
  for which the last step is $\ASM{Undo}$ together with releasing
  corresponding locks.

\end{itemize}

Second, the specification of $TA({\cal M},\ASM{TaCtl})$ leaves
deliberately open, when a machine $M \in \mathcal{M}$ will be started,
i.e., register as a transaction in $TransAct$ to be controlled by
\ASM{TaCtl}. This is in line with the common view that transactions $M
\in \mathcal{M}$ can register at any time to the transaction
controller \ASM{TaCtl} and will remain under its control until they
commit.

\begin{definition}\rm

\ Two runs $S_0,S_1,S_2,\dots$ and
$S_0^{\prime},S_1^{\prime},S_2^{\prime},\dots$ of $TA({\cal
  M},\ASM{TaCtl})$ are {\em equivalent} iff for each $M \in
\mathcal{M}$ the cleansed schedules   

\[\Delta_{i_1}(M),\Gamma_{i_1}(M),\Delta_{i_2}(M),\Gamma_{i_2}(M), \dots\] 

and

\[\Delta_{j_1}^\prime(M),\Gamma_{j_1}^\prime(M),
\Delta_{j_2}^\prime(M), \Gamma_{j_2}^\prime(M),\dots\] 

for the two runs are the same and the read locations and the values
read by~$M$ in $S_{i_k}$ and $S_{j_k}'$ are the same.

\end{definition}

That is, we consider runs to be equivalent, if all 
transactions $M \in \mathcal{M}$ read the same locations and see there
the same values and perform the same updates in the same order
disregarding waiting times and updates that are undone.

\subsection{Definition of serializability} 

Next we have to clarify our generalised notion of a serial run, for
which we concentrate on committed transactions -- transactions that
have not yet committed can still undo their updates, so they must be
left out of consideration\footnote{Alternatively, we could concentrate
  on complete, infinite runs, in which only committed transactions
  occur, as eventually every transaction will commit -- provided that
  fairness can be achieved.}. As stated above $\ASM{Abort}$ed
transactions are assumed to be eliminated from the run right at the
beginning. We need a definition of the read- and write-locations of
$M$ in a state $S$, i.e. $ReadLoc(M,S)$ and $WriteLoc(M,S)$ as used in
the definition of $newLocks(M,S)$.

We define $ReadLoc(M,S) = ReadLoc(r,S)$ and analogously
$WriteLoc(M,S)$ $= WriteLoc(r,S)$, where $r$ is the defining rule of
the ASM $M$. Then we use structural induction according to the
definition of ASM rules in ~\cite[Table 2.2]{BoeSta03}. As an
auxiliary concept we need to define inductively the read and write
locations of terms and formulae. The definitions use an
interpretation~$I$ of free variables which we suppress notationally
(unless otherwise stated) and assume to be given with (as environment
of) the state~$S$. This allows us to write $ReadLoc(M,S)$,
$WriteLoc(M,S)$ instead of $ReadLoc(M,S,I)$, $ReadLoc(M,S,I)$
respectively.

\subsubsection{Read/Write Locations of Terms and Formulae.}

\begin{asm}
ReadLoc(x,S)= WriteLoc(x,S)= \emptyset \mbox{ for variables }x \\
ReadLoc(f(t_1, \ldots ,  t_n) ,S)=\+
           \{(f,(eval(t_1,S), \ldots , eval(t_n,S)))\} 
                 ~\cup ~\bigcup_{1 \leq i \leq n}ReadLoc(t_i,S)\-
WriteLoc(f(t_1 , \ldots , t_n),S)=\{(f,(eval(t_1,S), \ldots , eval(t_n,S)))\}
\end{asm}
Logical variables (to be distinguished from programming variables
which are treated in the ASM framework as 0-ary functions and thus
stand for locations) appear in the $\LET$, $\FORALL$, $\CHOOSE$
constructs and are not locations: they cannot be written and their values
are not stored in a location but in the given interpretation~$I$ from
where they can be retrieved.

We define $WriteLoc(\alpha,S)=\emptyset$ for every formula $\alpha$
because formulae are not locations one could write into.
$ReadLoc(\alpha,S)$ for atomic formulae $P(t_1 , \ldots , t_n )$ has to be
defined as for terms with $P$ playing the same role as a function
symbol~$f$. For propositional formulae one reads the locations of
their subformulae. In the inductive step for quantified formulae
$domain(S)$ denotes the superuniverse of~$S$
minus the Reserve set~\cite[Ch.2.4.4]{BoeSta03} and 
$I_{x}^{d}$ the extension (or modification) of~$I$ where~$x$ is
interpreted by a domain element~$d$.

\begin{asm}
ReadLoc(P(t_1 , \ldots , t_n ),S)= \+
   \{(P,(eval(t_1,S), \ldots , eval(t_n,S)))\} 
                   ~\cup ~ \bigcup_{1 \leq i \leq n}ReadLoc(t_i,S)\-
ReadLoc(\neg \alpha)= ReadLoc(\alpha)\\
ReadLoc(\alpha_1 \wedge \alpha_2) = ReadLoc(\alpha_1) \cup ReadLoc(\alpha_2)\\
ReadLoc(\forall x \alpha,S,I)= \bigcup_{d \in domain(S)}ReadLoc(\alpha,S,I_{x}^{d})
\end{asm}
\noindent The values of the logical variables are not read
from a location but from the modified state environment function~$I_{x}^{d}$.

\subsubsection{Read/Write Locations of ASM Rules.}

\begin{asm}
ReadLoc(\SKIP,S)=WriteLoc(\SKIP,S)= \emptyset\\
ReadLoc(t_1 := t_2,S)= ReadLoc(t_1 :=_{op} t_2,S)=\+
          ReadLoc(t_1,S) \cup ReadLoc(t_2,S)\-
WriteLoc(t_1 := t_2,S)= WriteLoc(t_1 :=_{op} t_2,S)= WriteLoc(t_1,S)\\
ReadLoc(\IF \alpha \THEN r_1 \ELSE r_2,S) = \+
     ReadLoc(\alpha,S) \cup \left\{ 
           \begin{array}{ll}
             ReadLoc(r_1,S) & \IF eval(\alpha,S)=true \\
             ReadLoc(r_2,S) & \ELSE 
            \end{array} \right.\-
WriteLoc(\IF \alpha \THEN r_1 \ELSE r_2,S) =
          \left\{ 
           \begin{array}{ll}
             WriteLoc(r_1,S) & \IF eval(\alpha,S)=true\\
             WriteLoc(r_2,S) & \ELSE 
           \end{array} \right. \\
ReadLoc(\LET x = t \IN r,S,I)= ReadLoc(t,S,I) \cup ReadLoc(r,S,I_x^{eval(t,S)})\\
WriteLoc(\LET x = t \IN r,S,I)= WriteLoc(r,S,I_x^{eval(t,S)}) \mbox{ // call by value}\\
ReadLoc(\FORALL x \WITH \alpha \DO r,S,I)=  \+
     ReadLoc(\forall x  \alpha,S,I) ~\cup ~ 
                \bigcup_{a \in range(x,\alpha,S,I)}ReadLoc(r,S,I_x^a)\+
     \WHERE range(x,\alpha,S,I)= \{d \in domain(S) \mid 
                eval(\alpha,(S,I_x^d)) = true\}\dec\-
WriteLoc(\FORALL x \WITH \alpha \DO r,S,I)=  
     \bigcup_{a \in range(x,\alpha,S,I)}WriteLoc(r,S,I_x^a)\\
\end{asm}
In the following cases the same scheme applies to read and write
locations:\footnote{In $yields(r_1,S,I,U)$~$U$ denotes the update set
  produced by rule~$r_1$ in state~$S$ under~$I$.}
\begin{asm}
Read[Write]Loc(r_1 \PAR r_2,S)= \+
           Read[Write]Loc(r_1,S) \cup Read[Write]Loc(r_2,S)\-
Read[Write]Loc(r(t_1, \ldots, t_n),S) = Read[Write]Loc(P(x_1/t_1,\ldots,x_n/t_n),S)  \+
     \WHERE r(x_1,\ldots,x_n)=P \mbox{ // call by reference}\-
Read[Write]Loc(r_1 \SEQ r_2,S,I) = Read[Write]Loc(r_1,S,I) \cup \+
    \left\{ 
           \begin{array}{ll}
            Read[Write]Loc(r_2,S+U,I) & \IF
            yields(r_1,S,I,U) \AND Consistent(U)\\  
            \emptyset & \ELSE
     \end{array} \right.
\end{asm}
Due to the assumption that for component machines $M \in TransAct$
nondeterminism is expressed by choice functions no further clause is
needed to define $ReadLoc$ and $WriteLocs$ for machines of form
$\CHOOSE x \WITH \alpha \DO r$.

We say that~$M$ has or is committed (in state~$S_i$, denoted
$Committed(M,S_i)$) if step $\ASM{Commit}(M)$ has been performed (in
state~$S_i$).

\begin{definition}\rm

\ A run of $TA({\cal M},\ASM{TaCtl})$ is {\em serial} iff there is a
total order $<$ on $\mathcal{M}$ such that the following two
conditions are satisfied:

\begin{enumerate}

\item If in a state $M$ has committed, but $M^\prime$ has not, then
  $M < M^\prime$ holds.

\item If $M$ has committed in state $S_i$ and $M < M^\prime$ holds,
  then the cleansed schedule
\[ \Delta_{j_1}(M^\prime),\Gamma_{j_1}(M^\prime),
  \Delta_{j_2}(M^\prime), \Gamma_{j_2}(M^\prime),\dots \]

of $M^\prime$ satisfies $i < j_1$.

\end{enumerate}

\end{definition}

That is, in a serial run all committed transactions are executed
in a total order and are followed by the updates of transactions that have not
yet committed.

\begin{definition}\rm

\ A run of $TA({\cal M},\ASM{TaCtl})$ is {\em serialisable} iff it is
equivalent to a serial run of $TA({\cal
  M},\ASM{TaCtl})$.\footnote{Modulo the fact that ASM steps permit
  simultaneous updates of multiple locations, for ASMs with only
  genuine updates this definition of serializability is equivalent to
  Lamport's sequential consistency concept~\cite{Lamport79}.}

\end{definition}

\subsection{Serializability Proof} 

\begin{theorem}

\ Each run of $TA({\cal M},\ASM{TaCtl})$ is serialisable.

\end{theorem}

\begin{proof}
\ Let $S_0,S_1,S_2,\dots$ be a run of $TA({\cal M},\ASM{TaCtl})$. To
construct an equivalent serial run let $M_1 \in \mathcal{M}$ be a
machine that commits first in this run, i.e. $Committed(M,S_i)$
holds for some $i$ and whenever $Committed(M,S_j)$ holds for some $M
\in \mathcal{M}$, then $i \le j$ holds. If there is more than one
machine $M_1$ with this property, we randomly choose one of them.

Take the run of $TA(\{ M_1 \},\ASM{TaCtl})$ starting in state $S_0$,
say $S_0, S_1^\prime, S_2^\prime, \dots, S_n^\prime$. As $M_1$
commits, this run is finite. $M_1$ has been $\ASM{Delete}$d from
$TransAct$ and none of the $\ASM{TaCtl}$ components is triggered any
more: neither $\ASM{Commit}$ nor $\ASM{LockHandler}$ because
$CommitRequest$ resp. $LockRequest$ remain empty; not
$\ASM{DeadlockHandler}$ because $Deadlock$ remains false since~$M_1$
never $Wait$s for any machine; not $\ASM{Recovery}$ because$Victim$
remains empty. Note that in this run the schedule for $M_1$ is already
cleansed.

We now define a run $S_0^{\prime\prime} , S_1^{\prime\prime},
S_2^{\prime\prime}, \dots$ (of $TA({\cal M} - \{ M_1 \},\ASM{TaCtl})$,
as has to be shown) which starts in the final state $S_n^\prime =
S_0^{\prime\prime}$ of the $TA(\{ M_1 \},\ASM{TaCtl})$ run and where
we remove from the run defined by the cleansed schedules
$\Delta_i(M),\Gamma_i(M)$ for the originally given run all updates
made by steps of~$M_1$ and all updates in $\ASM{TaCtl}$ steps which
concern~$M_1$ (i.e. which are related to a lock or commit request by
$M_1$ or a victimization of $M_1$ or a $\ASM{TryToRecover}(M_1)$
step). Let
\[ \Delta_i^{\prime\prime} = \bigcup\limits_{M \in \mathcal{M} - \{M_1\} } \Delta_i(M) \cup 
\{ (\ell,v) \in \Delta_i(\ASM{TaCtl}) \mid 
          (\ell,v) \;\text{does not concern $M_1$} \} , \]
\[ \Gamma_i^{\prime\prime} = \bigcup^+\limits_{M \in \mathcal{M} - \{M_1\} } \Gamma_i(M) \cup^+ 
\{ (\ell,v) \in \Gamma_i(\ASM{TaCtl}) \mid 
          (\ell,v) \;\text{does not concern $M_1$} \} . \]

That is, in $\Delta_i^{\prime\prime},\Gamma_i^{\prime\prime}$ all updates are removed from the
original run which are done by $M_1$---their effect is reflected
already in the initial run segment from $S_0$ to $S_n^\prime $---or
are $\ASM{LockHandler}$ updates involving a $LockRequest(M_1)$ or
are $Victim(M_1):=true$ updates of the $\ASM{DeadlockHandler}$ or are
updates involving a $\ASM{TryToRecover}(M_1)$ step or are done by
a step involving a $\ASM{Commit}(M_1)$.

\begin{lemma}\label{lem1}

\ $S_0^{\prime\prime}, S_1^{\prime\prime}, S_2^{\prime\prime}, \dots$
is a run of $TA({\cal M} - \{ M_1 \},\ASM{TaCtl})$.

\end{lemma}

\begin{lemma}\label{lem2}

\ The run $S_0, S_1^\prime, S_2^\prime, \dots, S_n^\prime,
S_1^{\prime\prime}, S_2^{\prime\prime}, \dots$ of $TA({\cal
  M},\ASM{TaCtl})$ is equivalent to the original run
$S_0,S_1,S_2,\dots$.

\end{lemma}

By induction hypothesis $S_0^{\prime\prime}, S_1^{\prime\prime},
S_2^{\prime\prime}, \dots$ is serialisable, so
$S_0,S_1^\prime,S_2^\prime,\dots$ and thereby also $S_0,S_1,S_2,\dots$ is
serialisable with $M_1 < M$ for all $M \in \mathcal{M} - \{ M_1 \}$.\hfill
\end{proof}

\begin{proof}\textbf{(Lemma \ref{lem1})}
\ Omitting in $\Delta_i^{\prime\prime},\Gamma_i^{\prime\prime}$ from
$\Delta_i(\ASM{TaCtl}),\Gamma_i(\ASM{TaCtl})$ every update which
concerns $M_1$ leaves updates by $\ASM{TaCtl}$ in $S_i^{\prime\prime}$
concerning $M \neq M_1$.

It remains to show that every $\ASM{Fire}(M)$-step defined by
$\Delta_i(M),\Gamma_i(M)$ is a possible $\ASM{Fire}(M)$-step via
$\Delta_i^{\prime\prime},\Gamma_i^{\prime\prime}$ in a $TA({\cal M} -
\{ M_1 \},\ASM{TaCtl})$ run starting in $S_0''$.  Since the considered
$M$-schedule $\Delta_i(M),\Gamma_i(M)$ is cleansed, we only have to consider any
proper update step of $M$ in state~$S_i''$ (together with its
preceding lock request step, if any).

Case 1. $M$ for its steps uses $newLocks$ and some of the to-be-locked
locations are also $LockedBy(M_1)$.

Case 1.1. Some of the $newLocks$ granted to~$M$ are incompatible with
some of the locks granted to $M_1$. Then due to cleansing the
$newLocks$ are requestedd by~$M$ after $\ASM{Commit}(M_1)$ so that the
lock race between~$M$ and~$M_1$ is eliminated.

Case 1.2. The $newLocks$ granted to~$M$ are compatible with all locks
granted to $M_1$. The compatibility permits to shift the considered
proper $M$-step to after the next proper $M_1$-step.

Case 2. $M$ for its step uses $newLocks$ for locations but none of
them is $LockedBy(M_1)$. Then this~$M$-step can be shifted like in
Case 1.2.

Case 3.  $M$ for its step needs no $newLocks$. Then all needed locks
have been granted before and to those preceding steps the argument for
Case 1 and Case 2 applies by induction. \hfill
\end{proof}

\begin{proof}\textbf{(Lemma \ref{lem2})}
\ The cleansed machine schedules in the two runs, the read locations
and the values read there have to be shown to be the same.  First
consider any $M \not = M_1$. Since in the initial segment $S_0,
S_1^\prime, S_2^\prime, \dots, S_n^\prime$ no such~$M$ makes any move
so that its update sets in this computation segment are empty, in the
cleansed schedule of~$M$ for the run $S_0, S_1^\prime, S_2^\prime,
\dots, S_n^\prime, S_1^{\prime\prime}, S_2^{\prime\prime}, \dots$ all
these empty update sets disappear. Thus this cleansed schedule is the
same as the cleansed schedule of $M$ for the run
$S_n^\prime,S_1^{\prime\prime}, S_2^{\prime\prime}, \dots$ and
therefore by definition of $\Delta_i^{\prime\prime}(M) = \Delta_i(M)$
and $\Gamma_i^{\prime\prime}(M) = \Gamma_i(M)$ also for the original
run $S_0,S_1,S_2,\dots$ with same read locations and same values read
there.

Now consider $M_1$, its schedule 

\[\Delta_0(M_1),\Gamma_0(M_1),\Delta_1(M_1), \Gamma_1(M_1),\dots\]

 for the run $S_0,S_1,S_2,\dots$ and the corresponding cleansed
 schedule 

\[\Delta_{i_0}(M_1), \Gamma_{i_0}(M_1), \Delta_{i_1}(M_1), \Gamma_{i_1}(M_1), \dots\]. 

We proceed by induction on the cleansed schedule steps of $M_1$.  When
$M_1$ makes its first step using the updates in
$\Delta_{i_0}(M_1),\Gamma_{i_0}(M_1)$, this can only be a
$\ASM{Fire}(M_1)$-step together with the corresponding
$\ASM{RecoveryRecord}$ updates (or a lock request directly preceding
such a $\Delta_{i_1}(M_1),\Gamma_{i_1}(M_1)$-step) because in the
computation with cleansed schedule each lock request of $M_1$ is
granted and $M_1$ is not $Victim$ized. The values $M_1$ reads or
writes in this step have not been affected by a preceding incompatible
step of any $M \not = M_1$---otherwise~$M$ would have locked before
the corresponding locations and keep the locks until it commits (since
cleansed schedules are without $\ASM{Undo}$ steps), preventing $M_1$
from getting these locks which contradicts the fact that $M_1$ is the
first machine to commit and thus the first one to get the
locks. Therefore the values $M_1$ reads or writes in the step defined
by $\Delta_{i_0}(M_1),\Gamma_{i_0}(M_1)$ (resp. also
$\Delta_{i_1}(M_1),\Gamma_{i_1}(M_1)$) coincide with the corresponding
location values in the first (resp. also second) step of $M_1$
following the cleansed schedule with the same compatible updates of
partners of~$M$ to pass from $S_0$ to $S_1^\prime$ (case without
request of $newLocks$) resp. from $S_0$ to $S_1^\prime$ to
$S_2^\prime$ (otherwise). The shared updates of~$M_1$ are the same in
both runs by definition. The same argument applies in the inductive
step which establishes the claim.\hfill
\end{proof}

\section{Conclusion}

In this article we specified in terms of Abstract State Machines a
multi-level transaction controller $\ASM{TaCtl}$ and a multi-level
transaction operator which turns the behaviour of a set of concurrent
programs into a transactional one under the control of
$\ASM{TaCtl}$. In this way the locations shared by the programs and
possibly containing complex (hierarchically structured) values are
accessed in a well-defined manner. For this we proved that all
concurrent transactional runs are serialisable.

The relevance of the transaction operator is that it permits to
concentrate on the specification of program behavior ignoring any
problems resulting from the use of shared possibly nested
locations. That is, specifications can be written in a way that shared
locations, including those which contain complex values, are treated
as if they were exclusively used by a single program. This is valuable
for numerous applications, as shared locations (in particular
in a database with complex values) are common, and random access to them is
hardly ever permitted.

Furthermore, by shifting transaction control into the rigorous
framework of Abstract State Machines we made several extensions to
transaction control as known from the area of databases
\cite{elmasri:2006}. In the classical theory schedules are sequences
containing read- and write-operations of the transactions plus the
corresponding read- and write-lock and commit events, i.e., only one
such operation or event is treated at a time. In our case we exploited
the inherent parallelism in ASM runs, so we always considered an
arbitrary update set with usually many updates at the same time. Under
these circumstances we generalised the notion of schedule and
serialisability in terms of the synchronous parallelism of ASMs. More
importantly we included also partial updates to cope with (a
generalization of) multi-level transactions. In this way we stimulate
more parallelism in transactional systems. We were also able to
strengthen the multi-level transaction model by adding further
clarification about the dependencies between the levels---actually, we
showed that a strict organisation into levels is not required, as long
as subsumption dependencies are taken into consideration---and about
the necessity to provide inverse operators for the partial updates
that are used for higher-level operations.

Among further work we would like to be undertaken is to provide for
the transaction controller and the $TA$ operator specified in this
paper a (proven to be correct) implementation, in particular as
plug-in for the CoreASM~\cite{CoreASM,FaGeGl06} or
Asmeta~\cite{AsmMwebsite} interpreter engines. This needs in
particular a careful analysis of the subsumtion criterion. Note
however that the update instruction set concept in CoreASM realizes
the concept of partial updates as used here and defined
in~\cite{SchWan10}. We would also like to see refinements or
adaptations of our transaction controller model for different
approaches to serialisability~\cite{gray:1993}, to multi-level
transaction protocols~\cite{KiScZh09} and to other approaches to
transaction handling,
e.g. \cite{elmasri:2006,gray:1993,KiScZh09,ScRiDr00}. Last but not
least we would like to see further detailings of our correctness proof
to a mechanically verified one, e.g. using the ASM theories developed
in KIV (see~\cite{kiv} for an extensive list of relevant publications)
and PVS~\cite{GarRic00,GoVhLa96,Verifix96} or the
(Event-~\cite{Abrial10}) B~\cite{Abrial96} theorem prover for an
(Event-) B transformation of $TA({\cal M},\ASM{TaCtl})$ (as suggested
in~\cite{GlaHLR13}).

\paragraph{Acknowledgement} 
The attempt to define a plug-in transaction controller concept one can
apply to introduce sequential components into a concurrent computation
under transactional constraints was partly motivated by the ATM case
study presented to the Dagstuhl seminar reported
in~\cite{GlaHLR13}. The approach of this paper has been used for the
specification of the ASM modeling the ATM in~\cite{BoeZen15}. We thank
an anonymous referee for his insightful reading of our paper and in
particular for having pointed out a flaw in the original manuscript.
\def\note#1{}

\begin{appendix}

\section{Partial Updates}

\renewcommand{\thesection}{\Alph{section}}

The problem of partial updates in ASMs was first studied by Gurevich and Tillmann \cite{Gurevich-PartialUpdates01}. They observed that partial updates naturally arise in the context of synchronous parallel systems, and the problem has also manifested itself in the development of AsmL, an ASM-based specification language \cite{GurevichASML}. Although in principle partial updates can be avoided in the traditional ASM setting
     by explicitly formulating all intended partial updates to a structure by genuine updates,
     this can turn out to become rather cumbersome and in fact, AsmL
     required a solution that allows a programmer to freely use partial updates to modify counters, sets and maps in the main program and in submachines and in submachines of submachines, etc. without worrying
how submachines will report modifications and how to integrate modifications.
Therefore, Gurevich and Tillmann studied the problem of partial updates over the data types \emph{counter}, \emph{set} and \emph{map} \cite{Gurevich-PartialUpdates01}.  To develop a systematical approach for partial updates, they proposed an algebraic framework in which a \emph{particle} was defined as an unary modification operation over a
data type, and a parallel composition of particles as an abstraction of order-independent sequential composition. In doing so, they defined a partial update as a pair $(l, p)$ where $l$ is a location as in the standard ASMs, but $p$ is a particle which is different from a value $v$ in a genuine update $(l, v)$ of the traditional ASMs.

Nonetheless, Gurevich and Tillmann later realised that the previous framework was too limited, for example,  it failed to address partial updates over the data types
\emph{sequence} or \emph{labeled ordered trees} as exemplified in \cite{GurevichPartialUpdate03}.
This limitation led to the formalisation of {\em applicative algebras} as
a general solution to partial updates
\cite{gurevich:tcs2005}. It was shown that the problem of partial
updates over \emph{sequences} and \emph{labeled ordered trees} can
be solved in this more general algebraic framework, and the approach in
\cite{Gurevich-PartialUpdates01} was a special kind of an
applicative algebra.

\begin{definition}\label{def-applicative-algebra}\rm

\ An \emph{applicative algebra} consists of: (i) elements of a data type $\tau$, which include a \emph{trivial element} $\bot$ and at least one additional element, (ii) a monoid of particles over $\tau$ which include a \emph{null particle} $\lambda$ and the identity operation $id$, and (iii) a parallel composition $\Omega$ that, given an arbitrary finite multiset of particles, produces a particle. Each applicative algebra also needs to satisfy the following two conditions:

\begin{description}

\item[(1)] $p(\bot) = \bot$ for each particle $p$, and $\lambda(x) = \bot$ for every element $x$.

\item[(2)] $\Omega(\{\!\!\{p\}\!\!\}) =p$, $\Omega(M + \{\!\!\{id\}\!\!\}) = \Omega M$, and $\Omega (M + \{\!\!\{\lambda\}\!\!\}) = \lambda$.

\end{description}
A multiset $M$ of particles is called \emph{consistent} iff $\Omega M \neq \lambda$.

\end{definition}

Although applicative algebra provides a general framework for partial updates, the notion of particle is nonetheless not intuitive. Furthermore, the notion of location considered in these studies is the same as in the standard ASMs, which did not consider the subsumption relation between locations. Thus, the following definitions for partial updated were proposed in \cite{SchWan10}:

\begin{definition}\rm

A location $l_1$ {\em subsumes} a location $l_2$ if, for all states $S$, $eval(l_1,S)$
uniquely determines $eval(l_2,S)$.

\end{definition}

\begin{definition}\label{def-share_update}\rm

A \emph{partial update} is a triple
$(l, v, op)$ consisting of a location $l$, a value $v$, and a binary operator
$op$. Given a state $S$ and a single partial update $(l, v, op)$, we obtain a new state $S'$ by applying the partial update $(l, v, op)$ over $S$ and $eval(\ell, S')=op(eval(l,S),v)$.

\end{definition}

In the above definition, locations may subsume one another, i.e. one location is a substructure of another location. Intuitively, for a partial update $(l, v, op)$, the binary operator $op$ specifies how the value
$v$ partially affects the value of $l$ in the current state. When multiple partial updates are generated
to the same location simultaneously, a multiset $P_l$ of partial updates is obtained for the location $l$. The following definition of operator-compatibility ensures that partial updates to the same location are consistent in an update multiset.

\begin{definition}\label{C4:DefnCompatible}\rm

Let $P_l=\{\!\!\{(l,v_i,op_i) \mid i=1,...,k
\}\!\!\}$ be a multiset of partial updates to the same location $l$. Then $P_l$ is said to be \emph{operator-compatible} if,
for any two permutations $(\sigma_1,...,\sigma_k)$ and $(\pi_1,...,\pi_k)$ of $\{1,\dots, k\}$, we have the following for all $x$:

\medskip

\noindent $op_{\sigma_k}(...op_{\sigma_2}(op_{\sigma_1}(x, v_{\sigma_1}),v_{\sigma_2}),...,v_{\sigma_k}) =
op_{\pi_k}(...op_{\pi_2}(op_{\pi_1}(x, v_{\pi_1}),v_{\pi_2}),...,v_{\pi_k})$  

\medskip
\noindent An update multiset $P_l$ is \emph{consistent} if it is operator-compatible. 

\end{definition}

Based on the above definition, the following proposition is straightforward since, for an operator-compatible update multiset to the same location, applying its partial updates in any order yields the same result.

\begin{proposition}

If an update multiset $P_l$ is operator-compatible, then an order-independent sequential composition $\Theta$ of the partial update operations in $P_l$ (written as $\Theta P_l$) is equivalent to applying all the partial updates in $P_l$ sequentially in any order. That is, $\Theta P_l (x)=op_{\sigma_{|P_l|}}(...op_{\sigma_2}(op_{\sigma_1}(x, v_{\sigma_1}),v_{\sigma_2}),...,v_{\sigma_{|P_l|}})$ for any permutation $(\sigma_1,...,\sigma_{|P_l|})$ of $\{1,\dots, |P_l|\}$.

\end{proposition}

Therefore, if an update multiset $P_l$ is consistent, then all the partial updates in $P_l$ can be aggregated into one genuine update on the same location $l$.
 A state $S'$ can be obtained from $S$ by applying the multiset $P_l$ of partial updates sequentially, and we have $eval(l, S')=\Theta P_l(eval(l,S))$.

\end{appendix}

\bibliographystyle{abbrv}
\bibliography{partialupdate,TransactionAsmbib}

\end{document}